\documentclass[10pt]{amsart}
\usepackage{cmap}
\usepackage{amsmath}
\usepackage{amssymb}
\usepackage{amsfonts}

\usepackage{pstricks,pst-node,cite,epsfig}
\usepackage{gastex}
\usepackage{graphicx}
\usepackage{xcolor}
\usepackage{mathrsfs}
\usepackage{textcomp}
\usepackage{setspace}
\usepackage{array}
\usepackage{float}
\usepackage{fancyhdr}
\usepackage{url}
\usepackage{pgf,tikz}
\usetikzlibrary{arrows,automata,positioning,decorations.pathreplacing,decorations.markings}
\usepackage{cite,tabularx,color}
\usepackage{tcolorbox}
\usepackage{caption}

\newtheorem{thm}{Theorem}

\newtheorem{pro}[thm]{Proposition}
\newtheorem{rmk}[thm]{Remark}

\newtheorem{lem}[thm]{Lemma}

\newcommand{\ovl}{\overline}

\renewcommand{\refname}{References}
\renewcommand{\proofname}{Proof.}

\begin{document}
\renewcommand{\refname}{References}
\renewcommand{\proofname}{Proof.}
\thispagestyle{empty}

\title[Sat solvers for synchronization in non-deterministic automata]{Using Sat solvers for synchronization issues\\
in non-deterministic automata}
\author{Hanan Shabana and Mikhail V. Volkov}%
\address{Hanan Shabana
\newline\hphantom{iiii}	 Institute of Natural Sciences and Mathematics, Ural Federal University
\newline\hphantom{iiii}   Lenina 51, 620000 Yekaterinburg, Russia
\newline\hphantom{iiii}	 Faculty of Electronic Engineering, Menoufia University, Egypt}%
\email{hananshabana22@gmail.com}%
\address{Mikhail V. Volkov
\newline\hphantom{iiii}	 Institute of Natural Sciences and Mathematics, Ural Federal University
\newline\hphantom{iiii}   Lenina 51, 620000 Yekaterinburg, Russia}%
\email{Mikhail.Volkov@usu.ru}%

\thanks{\rm Supported by the Russian Foundation for Basic Research, grant no.\ 16-01-00795, the Ministry of Education and Science of the Russian Federation, project no.\ 1.3253.2017, and the Competitiveness  Enhancement Program of Ural Federal University.}

\maketitle {\small
\begin{quote}
\noindent{\sc Abstract. } We approach the problem of computing a $D_{3}$-synchronizing word of minimum length for a given nondeterministic automaton via its encoding as an instance of SAT and invoking a SAT solver. We also present some experimental results.

\medskip

 \end{quote}
}

\section{Background and overview}
\label{sec:intro}
We assume the reader's familiarity with some basic concepts of computational complexity theory that can be found in the early chapters of any general complexity theory text such as, e.g., \cite{Papa}. As far as automata theory is concerned, we have tried to make the paper, to a reasonable extent, self-contained.

One of the significant concepts for digital systems is \emph{synchronization}. It means that all parts of the system are in agreement regarding the present state of the system. This concept is of immense importance in fields such as coding theory, conformance testing, biocomputing, industrial robotics, and many others, and also leads to intriguing mathematical questions, see, e.g., \cite{volkovservey}.

From the viewpoint of mathematics, discrete systems are often modeled as finite automata. A \emph{finite automaton} is a triple $\mathscr{A}=(Q,\Sigma,\delta)$, where $Q$ is a finite non-empty set which elements are referred to as \emph{states}, $\Sigma$ is a finite non-empty set which is called the \emph{input alphabet} and which elements are referred to as \emph{input symbols} or \emph{input letters}, and $\delta$ is a map, called the \emph{transition function}, that describes the action of symbols in $\Sigma$ at states in $Q$.  Finite automata are usually classified into three categories according to the nature of their transition function.

\begin{enumerate}
\item[DFA:] $\mathscr{A}=(Q,\Sigma,\delta)$ is a \emph{deterministic finite automaton} (DFA) if the transition function $\delta$ is a total map $Q\times\Sigma\rightarrow Q$, that is, $\delta (q,s)$  is defined for every  state $q\in Q$ and for every  symbol $s\in\Sigma$. We interpret  $\delta (q,s)$ as the next state where the DFA would move to if it was at the state $q$ and read the symbol $s$.

\item[PFA:] $\mathscr{A}=(Q,\Sigma,\delta)$ is a \emph{partial finite automaton} (PFA) if the transition function $\delta$ is a partial map $Q\times \Sigma\rightarrow Q$, that is, $\delta (q,s)$  is defined for some pairs $(q,s)\in Q\times\Sigma$ but may be undefined for some other pairs. We again interpret  $\delta (q,s)$, provided it is defined, as the next state where the PFA would move to if it was at the state $q$ and read the symbol $s$, and we write $\delta (q,s)=\varnothing$ to indicate that $\delta (q,s)$ is undefined\footnote{It should be noted that in the literature, automata that we call PFAs sometimes are referred to as deterministic finite automata while our DFAs are called \emph{complete} deterministic finite automata.}.

\item[NFA:] $\mathscr{A}=(Q,\Sigma,\delta)$ is a \emph{nondeterministic finite automaton} (NFA) if the  transition function $ \delta $ is a map $ Q \times \Sigma\rightarrow \mathcal{P}(Q)$, where $\mathcal{P}(Q)$ is the power set of $Q$, that is, for every  state $q\in Q$ and for every symbol $s\in\Sigma$, the expression $\delta (q,s)$ is not a single  state, but rather a subset of states. If this subset is non-empty, we interpret it as the set of all possible states where the NFA could move to if it was at the state $q$ and read the symbol $s$. If $\delta (q,s)=\varnothing$, we say that the action of $s$ is undefined at $q$.
\end{enumerate}

Clearly, both DFSs and PFAs can be considered as special instances of NFAs. Therefore, in the sequel, we define all concepts for NFAs, commenting on their specializations for NFAs and PFAs, if necessary.

We represent a given automaton $\mathscr{A}=(Q,\Sigma,\delta)$ by the labeled directed graph with the vertex set $Q$, the label alphabet $\Sigma$, and the set of labeled edges
\[
\{q\xrightarrow{s}q'\mid q,q'\in Q,\ s\in\Sigma,\ q'\in\delta(q,s)\}.
\]
Figure \ref{dfa&nfa} shows examples of a DFA (left) and a NFA (right). We adopt the convention that edges with multiple labels represent bunches of parallel edges. Thus, the edge $1\xrightarrow{a,c}0$ in Figure~\ref{dfa&nfa} represents the two parallel edges $1\xrightarrow{a}0$ and $1\xrightarrow{c}0$, etc.

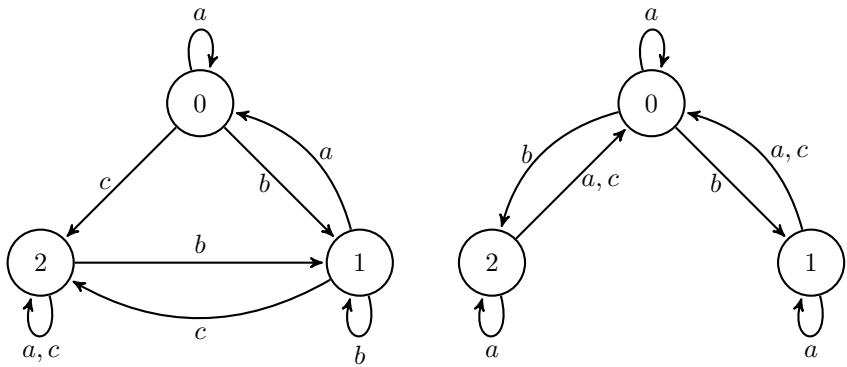
\begin{figure}[h]

\centering

\begin{tikzpicture}[->,>=stealth',shorten >=1pt,auto,node distance=3cm,
thick,main node/.style={circle,draw,font=\sffamily\Large\bfseries}]

  \node[state]         (0)         {$0$};
  \node[state]         (1) [below right of=0] {$1$};
  \node[state]         (2) [below left of=0] {$2$};

  \path (0) edge   [loop above]         node {$a$} (0)
            edge   node[left]           {$b$} (1)
            edge   [ right]             node [left] {$c$} (2)

        (1) edge   [loop below]         node {$b$} (1)
            edge   [bend right]         node[right] {$a$} (0)
            edge   [bend left]          node {$c$} (2)

        (2) edge   [loop below]           node {$a,c$} (2)
            edge                         node {$b$} (1);
  \node                (7) [right of=0] {};

  \node[state]         (3) [right of=7]        {$0$};
  \node[state]         (4) [below right of=3] {$1$};
  \node[state]         (5) [below left of=3] {$2$};

  \path (3) edge    [loop above]      node {$a$} (3)
            edge    node[left]        {$b$} (4)
            edge    [bend right]      node [left] {$b$} (5)

        (4) edge    [loop below]      node {$a$} (4)
            edge   [bend right]      node[right] {$a,c$} (3)

        (5) edge    [loop below]       node {$a$} (5)
            edge    node[right]        {$a,c$} (3);
\end{tikzpicture}
\caption{A DFA (left) and a NFA (right) with $Q=\{0,1,2\}$ and $\Sigma=\{a,b,c\}$}
\label{dfa&nfa}
\end{figure}

Given an alphabet $\Sigma$, a \emph{word} over $\Sigma$ is a finite sequence of symbols from $\Sigma$. We do not exclude the empty sequence from this definition; that is, we allow the \emph{empty word}. The set of all words over $\Sigma$ including the empty word is denoted by $\Sigma^*$ and is referred to as the \emph{free monoid over $\Sigma$}. If $w=a_1\cdots a_\ell$ with $a_1,\dots,a_\ell\in\Sigma$ is a non-empty word over $\Sigma$, the number $\ell$ is said to be the \emph{length} of $w$ and is denoted by $|w|$. The length of the empty word is defined to be 0. The set of all words of a given length $\ell$ over $\Sigma$ is denoted by $\Sigma^\ell$.

For every NFA $\mathscr{A}=\langle Q,\Sigma,\delta\rangle$, the transition function $\delta$ can be extended to a function $\mathcal{P}(Q)\times\Sigma^{*}\rightarrow \mathcal{P}(Q)$ (still denoted by $\delta$) by induction on the length of $w\in\Sigma^*$. If $|w|=0$, that is, $w$ is the empty word,  then, for each $X\subseteq Q$, we let $\delta(X,w)=X$. If $|w|>0$, we represent $w$ as $w=sw'$ with $w'\in\Sigma ^{*}$ and $s\in \Sigma$ and, for each $X\subseteq Q$, let $\delta(X,w)=\bigcup_{q\in X}\delta(\delta(q,s),w')$. (The right hand side of the latter equality is defined by the induction assumption since $|w'|<|w|$.) To lighten the notation, we write $q.w$ for $\delta(q,w)$ and $X.w$ for $\delta(X,w)$ whenever we deal with a fixed automaton.

Here we are interested in \emph{synchronization} of finite automata. The idea of \emph{synchronization} is as follows: for a given automaton, we are looking for an input word that directs the  automaton to a specific state, no matter at which state the automaton was at the beginning. This input is called a \emph{synchronizing word}, and if an automaton  possesses such a  word, it is called \emph{synchronizing}.

The above informal idea of synchronization is easy to formalize for DFAs but for NFAs it admits several non-equivalent formalizations. First, we recall the three versions that were suggested in~\cite{4} and have been widely studied thereafter.

Let $\mathscr{A}=(Q,\Sigma,\delta)$ be an NFA, $i=1,2,3$. A word $w\in\Sigma^*$ is said to be $D_{i}$-\emph{syn\-chro\-nizing} for $\mathscr{A}$ if it satisfies the condition $(D_i)$ from the list below:
\begin{enumerate}
\item[$(D_{1})$:]  $\forall q\in Q$ $(q.w\ne\varnothing\ \wedge\ |q.w|=|Q.w|=1)$;
\item[$(D_{2})$:]  $\forall q\in Q$ $(q.w\ne\varnothing\ \wedge\ q.w=Q.w)$;
\item[$(D_{3})$:]  $\bigcap_{q\in Q}q.w\ne\varnothing$.
\end{enumerate}
A NFA is called $D_{i}$-\emph{synchronizing}, $i=1,2,3$, if it has a $D_{i}$-synchronizing word\footnote{In some sources, the requirement $q.w\ne\varnothing$ is not explicitly included in the definition of $D_2$-syn\-chronization. If one omits this requirement, every word that is nowhere defined becomes $D_{2}$-synchronizing. We think this version of synchronization hardly is of independent interest since it readily reduces to $D_2$-synchronization in our sense in the automaton obtained from $\mathscr{A}$ by adding a new sink state and making all transitions undefined in $\mathscr{A}$ lead to this sink state.}. It should be clear that every $D_{1}$-synchronizing word is also $D_{2}$-synchronizing and every $D_{2}$-synchronizing word is also $D_{3}$-synchronizing. The converse is not true in general. For an illustration, consider the NFA $\mathscr{A}$ in Figure \ref{dfa&nfa} (right). It is easy to see that for it, the word $abc$ is $D_{1}$-synchronizing, the word $ab$ is $D_{2}$-synchronizing, but not $D_{1}$-synchronizing, and the word $a$ is $D_{3}$-synchronizing, but not $D_{2}$-synchronizing. Moreover, the NFA obtained from $\mathscr{A}$ by omitting the letter $c$ is $D_{2}$-synchronizing, but not $D_{1}$-synchronizing, while the NFA obtained from $\mathscr{A}$ by omitting the letters $b$ and $c$ is $D_{3}$-synchronizing, but not $D_{2}$-synchronizing.

Yet another version of synchronization for NFAs has been studied by Martyugin, see, e.g., \cite{Martyugin}. Let $\mathscr{A}=(Q,\Sigma,\delta)$ be an NFA. A word $w=a_1\cdots a_\ell$ with $a_1,\dots,a_\ell\in\Sigma$ is said to be \emph{carefully synchronizing} for $\mathscr{A}$ if it satisfies the condition $(C)$, being the conjunction of $(C1)$--$(C3)$ below:
\begin{enumerate}
\item[$(C1)$:] $\delta(q,a_1)$ is defined for all $q\in Q$,
\item[$(C2)$:] $\delta(q,a_i)$ with $1<i\le \ell$ is defined for all $q\in Q.a_1\cdots a_{i-1}$,
\item[$(C3)$:] $|Q.w|=1$.
\end{enumerate}
Thus, when $w$ is applied at any state in $Q$, no undefined transition occurs during the course of application. Clearly, every carefully synchronizing word is also $D_{1}$-synchronizing but the converse is not true. For instance, the word $abc$ is not carefully synchronizing for the NFA $\mathscr{A}$ in Figure \ref{dfa&nfa} (right); moreover, this NFA possesses no carefully synchronizing word. We call a NFA \emph{carefully synchronizing} if it admits a carefully synchronizing word. Thus, if we denote by $\mathbf{D}_i$, $i=1,2,3$, the class of all $D_{i}$-synchronizing NFAs and by $\mathbf{C}$ the class of all carefully synchronizing NFAs, we have the following strict inclusions:
\[
\mathbf{C}\subset\mathbf{D}_1\subset \mathbf{D}_2\subset\mathbf{D}_3.
\]

In this paper, we consider $D_3$-synchronization. As it can been seen from the above discussion, it is the most general version of synchronization for NFAs amongst those considered in the literature so far. Besides that, we think that it reasonably reflects the basic nature of non-determinism. Indeed, if an NFA $\mathscr{A}=(Q,\Sigma,\delta)$ is used as an \emph{acceptor}, we designate some states in $Q$ as initial and final and then say that $\mathscr{A}$ \emph{accepts} a word $w\in\Sigma^*$ whenever there exists a path labeled $w$ that starts at an initial state and terminates at a final state. The definition of a $D_{3}$-syn\-chronizing word very much resembles this concept: a word $w\in\Sigma^*$ is $D_{3}$-syn\-chro\-nizing whenever for each $q\in Q$, there exists a path labeled $w$ that starts at $q$ and terminates at a certain common state, independent of $q$. In both cases we do not require that a starting state uniquely determines the path labeled $w$ nor that every path labeled $w$ with a given starting state should arrive at a final/common state.

We also mention in passing that $D_3$-synchronization gets a very transparent meaning within a standard matrix representation of NFAs. In this representation, an NFA $\mathscr{A}=(Q,\Sigma,\delta)$ becomes a collection of $|\Sigma|$ Boolean $Q\times Q$-matrices where to each input symbol $s\in\Sigma$, a  matrix $M(s)$ is assigned such that the $(q,q')$-entry of $M(s)$ is 1 if $q'\in\delta(q,s)$ and 0 otherwise. Then it is not hard to realize that the automaton $\mathscr{A}$ is $D_{3}$-synchronizing if and only if some product of the matrices $M(s)$, $s\in\Sigma$, has a column consisting entirely of 1s.

Some information about $D_3$-synchronization can be found in Chapter~8 of Ito's monograph~\cite{Ito}; recently, some aspects of $D_3$-synchronization has been considered in \cite{BJO15,GGJ16,DonZantema17,fuzzy}. (The papers \cite{BJO15,GGJ16} use the language of matrices rather than that of automata.)

It is easy to see that each of the conditions $(C)$, $(D_{1})$, $(D_{2})$, $(D_{3})$ leads to the same notion when restricted to PFAs. Thus, for PFAs and, in particular, for DFAs, we call a word \emph{synchronizing} if it satisfies any of these conditions. A PFA (in particular, a DFA) is said to be \emph{synchronizing} if it has a synchronizing word.

It is known that the problem of determining whether or not a DFA with $n$ states is synchronizing can be solved in $O(n^2)$ time, see, e.g., \cite{volkovservey} or \cite{Sandberg}. If such a DFA is synchronizing, it always has a synchronizing word of length $(n^3-n)/6$, see \cite{pin}, and it is conjectured that a synchronizing DFA with $n$ states must have a synchronizing word of length $(n-1)^{2}$ (this is the famous \v{C}ern\'{y} conjecture). In contrast, the problem of determining  whether or not a given PFA is synchronizing is known to be PSPACE-complete and there is no polynomial in $n$ upper bound on the length of synchronizing words for a synchronizing PFA with $n$ states. (These results were found by Rystsov in the early 1980s~\cite{Rystsov:1980,Rystsov:1983} and later rediscovered (and strengthened) by Martyugin~\cite{Martyugin:2012}.) This readily implies that the problem of determining whether or not a given NFA is $D_3$-synchronizing as well as the problem of finding a $D_3$-synchronizing word of minimum length are computationally hard.

Nowadays, a popular approach to computationally hard problems consists in encoding them as instances of the Boolean satisfiability problem (SAT) that are then fed to a SAT-solver, that is, a specialized program designed to solve instances of SAT. We refer to this approach as the \emph{SAT-solver method}. Modern SAT solvers can solve instances with hundreds of thousands of variables and millions of clauses within a few minutes. Thanks to this remarkable progress, the SAT-solver method has proved to be very efficient for an extremely  wide range of problems of both theoretical and practical importance. Its applications are far too numerous to be listed here; some examples of such applications can be found in the survey~\cite{HKR}, which also gives a smart introduction into the area. Here we mention only three recent papers that deal with two difficult problems related to finite automata. Geldenhuys, van der Merwe, and van Zijl \cite{africa} have used the SAT-solver method to attack the minimization problem for NFAs. In the minimization problem, which is known to be PSPACE-complete~\cite{JR91}, an NFA $\mathscr{A}$ with designated initial and final states is given, and one looks for an NFA of minimum size that accepts the same set of words as $\mathscr{A}$. Skvortsov and Tipikin \cite{Tipikin} have applied the method to find a synchronizing word of minimum length for a given DFA with two input symbols, and G\"uni\c{c}en, Erdem, and Yenig\"un \cite{GEY} have extended their approach to DFAs with arbitrary input alphabets. The problem of finding a synchronizing word of minimum length is known to be hard for the complexity class $\mathrm{FP}^\mathrm{NP[log]}$, the functional analogue of the class of problems solvable by a deterministic polynomial-time Turing machine that has an access to an oracle for an NP-complete problem, with the number of queries being logarithmic in the size of the input~\cite{Olschewski&Ummels:2010}.

In the present paper, we use the SAT-solver method to approach the problem of computing a $D_{3}$-synchronizing word of minimum length for a given NFA. It should be stressed that neither the encoding of NFAs used in~\cite{africa} nor the encoding of synchronization used in~\cite{Tipikin,GEY} work for our problem, and therefore, we have had to invent essentially different encodings.

The rest of the paper is divided into three sections. Section \ref{encoding} describes our basic encoding and Section \ref{experimental} presents implementation details and some of our experimental results. The final section contains several concluding remarks and a discussion of possible further developments.

\section{Encoding}
\label{encoding}

By the encoding of a problem, we mean a polynomial reduction from this problem to SAT. First, let us precisely formulate the problem which we are interested in.

\begin{tcolorbox}
D3W (the existence of a $D_{3}$-synchronizing word of a given length):

\textsc{Input}: A NFA $\mathscr{A}$  with two input symbols and  a positive integer $\ell$.

\textsc{Output}: YES if $\mathscr{A}$ has a $D_{3}$-synchronizing word of length $\ell$; NO otherwise.
\end{tcolorbox}

The integer $\ell$ is assumed to be given in unary. With $\ell$ given in binary, a polynomial reduction from D3W to SAT is hardly possible. Indeed, it is known that every $D_{3}$-synchronizing NFA with $n$ states has a $D_{3}$-synchronizing word of length at most $2^n$, see~\cite[Proposition 8.3.10]{Ito}. Hence, given a NFA  $\mathscr{A}$ with $n$ states and two input symbols, the answer to the problem D3W for the instance $(\mathscr{A},2^n)$ is YES if and only if $\mathscr{A}$ is $D_{3}$-synchronizing. As it was mentioned, the problem of determining whether or not a given NFA is $D_3$-synchronizing is PSPACE-complete, whence the version of D3W in which the integer parameter is given in binary is PSPACE-hard. On the other hand, SAT is an archetypical problem in NP, and clearly, the existence of a polynomial reduction from a PSPACE-hard problem to a problem in NP would imply that the polynomial hierarchy collapses at level~1. While, as it is usual in complexity theory, the question of whether or not the polynomial hierarchy collapses at any level is open, a common opinion is that it does not.

In contrast, the version of D3W with the integer parameter given in unary is easily seen to belong to NP. Indeed, given an  instance $(\mathscr{A},\ell)$ of D3W in this setting, one has right to guess a word $w$ of length $\ell$ over the input alphabet of $\mathscr{A}$ as $w$ is obviously of polynomial size in terms of the size of the instance.  Then one just checks whether or not $w$ is $D_{3}$-synchronizing for $\mathscr{A}$, and time spent for this check is clearly polynomial in the size of $(\mathscr{A},\ell)$. By Cook's classic theorem (see, e.g., \cite[Theorem 8.2]{Papa}), SAT is NP-complete, and by the very definition of NP-completeness, there exists a polynomial reduction from our version of D3W to SAT.

Recall that an instance of SAT is a pair $(V,C)$, where $V$ is a set of Boolean variables and $C$ is a collection of clauses over $V$. (A \emph{clause} over $V$ is a disjunction of literals and a \emph{literal} is either a variable in $V$ or the negation of a variable in~$V$.) Any \emph{truth assignment} on $V$, i.e., any map $\varphi\colon V\to\{0,1\}$, extends to a map $C\to\{0,1\}$ (still denoted by $\varphi$) via the usual rules of propositional calculus: $\varphi(\neg x)=1-\varphi(x)$, $\varphi(x\vee y)=\max\{\varphi(x),\varphi(y)\}$. A truth assignment $\varphi$ \emph{satisfies} $C$ if $\varphi(c)=1$ for all $c\in C$. The answer to an instance $(V,C)$ is YES if $(V,C)$ has a \emph{satisfying assignment} (i.e., a truth assignment on $V$ that satisfies $C$) and NO otherwise.

Thus, a polynomial reduction from D3W to SAT is an algorithm that, given an arbitrary instance $(\mathscr{A},\ell)$ of D3W, constructs, in polynomial time with respect to the size of $(\mathscr{A},\ell)$, an instance $(V,C)$ of SAT such that the answer to $(\mathscr{A},\ell)$ is YES if and only if so is the answer to $(V,C)$. Of course, neither a pure existence statement nor any general construction that can be extracted from one of the proofs of Cook's theorem can be used for our purposes. We need a sort of ``practical'' reduction: it should be explicit, easy to implement, and economical in the sense that the degrees of the polynomials that bound the number of variables in $V$ and the number of clauses in $C$ in terms of the size of $(\mathscr{A},\ell)$ should be as small as possible.

In the following presentation of our encoding, precise definitions and statements are interwoven with less formal comments explaining the ``physical'' meaning of variables and clauses we introduce and with estimations of their numbers.

So, take a NFA $\mathscr{A}=(Q,\Sigma,\delta)$ and an integer $\ell>0$. Denote the size of $Q$ by $n$ and fix some numbering of the states in $Q$ so that $Q=\{q_1,\dots,q_n\}$. Recall that we consider the problem D3W for NFAs with two input symbols, so let $\Sigma=\{0,1\}$.

We start with introducing the variables used in the instance $(V,C)$ of SAT that encodes $(\mathscr{A},\ell)$. The set $V$ consists of three sorts of variables:  \emph{letter variables}, \emph{token variables}, and \emph{synchronization variables}.

The letter variables are $x_1,\dots,x_\ell$. They are just placeholders for the input symbols 0 and 1. There is an obvious 1-1 correspondence between the truth assignments on the set $X=\{x_1,\dots,x_\ell\}$ and the words in $\Sigma^\ell$: given a truth assignment $\varphi\colon X\to\{0,1\}$, the corresponding word is $\varphi(x_1)\cdots\varphi(x_\ell)$, and, conversely, given a word $a_1\cdots a_\ell$ with $a_1,\dots,a_\ell\in\{0,1\}$, the corresponding truth assignment is $x_t\mapsto a_t$ for each $t=1,\dots,\ell$.

The token variables are $y_{ij}^t$ where $i,j=1,\dots,n$ and $t=0,1,\dots,\ell$. To explain the role of these variables, we use a solitaire-like game $\Gamma$ on the labeled directed graph representing the NFA $\mathscr{A}$. In the initial position of $\Gamma$, each state $q_i\in Q$ holds exactly one token denoted $\mathbf{i}$. In the course of the game, tokens migrate and may multiply or disappear according to certain rules that will be specified a bit later, when we describe the clauses in $C$. For the moment, it is sufficient to say that the rules are designed to ensure that the variable $y_{ij}^t$ gets value 1 in a satisfying truth assignment for $C$ if and only if after $t$ rounds of the game, one of the tokens held by the state $q_j$ is $\mathbf{i}$.

The synchronization variables are $z_1,\dots,z_n$. They play the role of indicators showing which states may occur at the end of the synchronization process. By the definition of $D_3$-synchronization, the answer to the instance $(\mathscr{A},\ell)$ is YES if and only if there exists a word $w\in\Sigma^\ell$ such that $\bigcap_{q\in Q}q.w\ne\varnothing$. The clauses of $C$ will be chosen so that the variable $z_j$ gets value 1 in a satisfying assignment for $C$ if and only the state $q_j$ belongs to the set $\bigcap_{q\in Q}q.w$, where $w$ is the word defined by the restriction of the assignment to $X$.

We see that the total number of variables in $V$ is $\ell+n^2(\ell+1)+n$.

Now we turn to constructing the set of clauses $C$. It is the disjoint union of $\ell+1$ sets: the set $C_0$ of \emph{initial clauses}, the sets $C_t$, $t=1,\dots,\ell$, of \emph{transition clauses}, and the set $S$ of \emph{synchronization clauses}.

The clauses in $C_0$ describes the initial position of our game $\Gamma$. As mentioned, in this position, each state $q_i\in Q$ holds the token $\mathbf{i}$ and nothing else. It order to reflect this setting, we let $C_0$ consist of the clauses $y_{11}^0,\dots,y_{nn}^0$ along with all clauses of the form $\neg y_{ij}^0$ with $i\ne j$. Altogether, $C_0$ contains $n^2$ one-literal clauses.

They are the clauses in $C_t$, $t=1,\dots,\ell$, that encode the rules of $\Gamma$. The rules are as follows. At each move an input symbol $a\in\Sigma$ is chosen. Then for each state $q\in Q$ such that $q.a\ne\varnothing$, all tokens that were held by $q$ slide along the edges labeled $a$ to all states in the set $q.a$. (If $|q.a|>1$, then every token held by $q$ multiplies to $|q.a|$ identical tokens, one for each state in $q.a$.) If $q.a=\varnothing$, then all tokens that were held by $q$ disappear. Thus, after the move, the token $\mathbf{i}$ occurs at a state $p\in Q$ if and only if $p\in q.a$ for some state $q$ that had held $\mathbf{i}$ just prior to the move.

For an illustration, Figure~\ref{moves} demonstrates the initial distribution of tokens on a 5-state NFA with the input alphabet $\{0,1\}$ (top), along with the outcomes of the first move, depending on whether 0 or 1 has been chosen for the move (bottom left and bottom right, respectively).

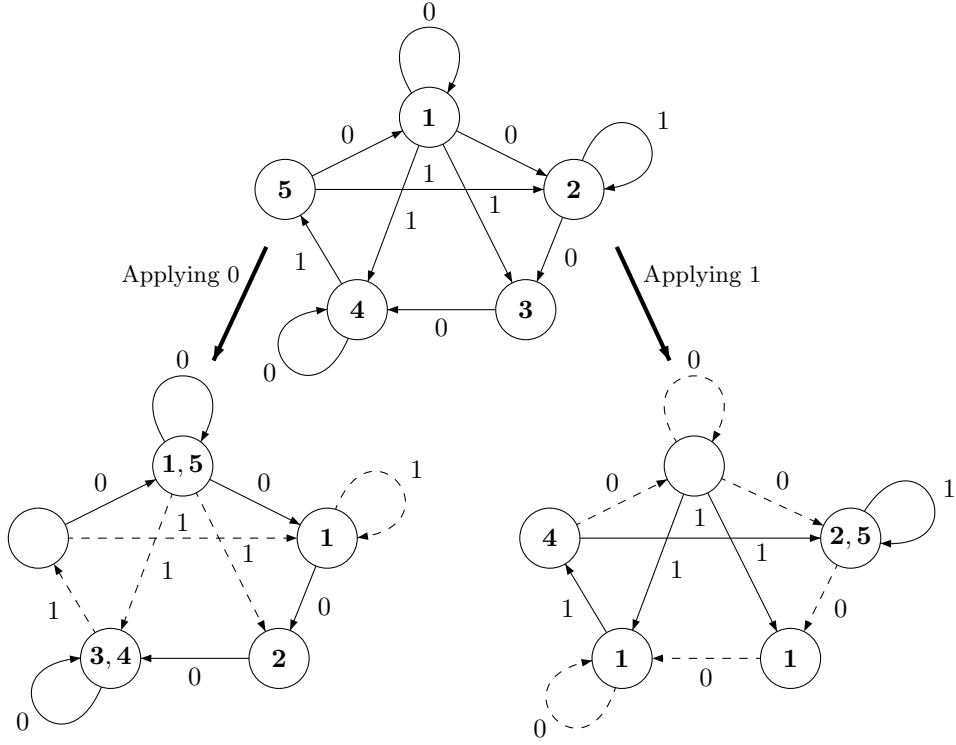
\begin{figure}[t]
\begin{center}
\unitlength=.8mm
\begin{picture}(166,125)(0,-55)
\node[Nframe=n,Nfill=n](b0)(46.0,33.0){}
\node[Nframe=n,Nfill=n](b1)(33.0,5.0){}
\drawedge[ELside=r,linewidth=.7,AHdist=2.41,%
AHangle=20,AHLength=2.5,AHlength=2.41](b0,b1){\small Applying $0$}
\node[Nframe=n,Nfill=n](b0)(100.0,33.0){}
\node[Nframe=n,Nfill=n](b1)(113.0,5.0){}
\drawedge[linewidth=.7,AHdist=2.41,AHangle=20,
AHLength=2.5,AHlength=2.41](b0,b1){\small Applying $1$}
\node(n10)(71.0,50.0){$\mathbf{1}$}
\node(n11)(47.0,38.0){$\mathbf{5}$}
\node(n12)(95.0,38.0){$\mathbf{2}$}
\node(n13)(59.0,18.0){$\mathbf{4}$}
\node(n14)(87.0,18.0){$\mathbf{3}$}
\drawedge[ELdist=1.77](n10,n12){$0$}
\drawedge[ELdist=2.33](n12,n14){$0$}
\drawedge[ELdist=1.59](n14,n13){$0$}
\drawedge[ELdist=2.24](n13,n11){$1$}
\drawedge[ELdist=1.83](n11,n10){$0$}
\drawedge[ELdist=1.19](n11,n12){$1$}
\drawedge[ELdist=1.71](n10,n14){$1$}
\drawedge[ELdist=1.71](n10,n13){$1$}
\drawloop[loopangle=90](n10){$0$}
\drawloop[loopangle=215.88](n13){$0$}
\drawloop[ELdist=2.08,loopangle=37.3](n12){$1$}
\node(n0)(30.0,-8.0){$\mathbf{1},\mathbf{5}$}
\node(n1)(6.0,-20.0){}
\node(n2)(54.0,-20.0){$\mathbf{1}$}
\node(n3)(18.0,-40.0){$\mathbf{3},\mathbf{4}$}
\node(n4)(46.0,-40.0){$\mathbf{2}$}
\node(n5)(115,-8.0){}
\node(n6)(91,-20.0){$\mathbf{4}$}
\node(n7)(141,-20.0){$\mathbf{2},\mathbf{5}$}
\node(n8)(103,-40.0){$\mathbf{1}$}
\node(n9)(131,-40.0){$\mathbf{1}$}
\drawedge[ELdist=1.77](n0,n2){$0$}
\drawedge[ELdist=2.33](n2,n4){$0$}
\drawedge[ELdist=1.59](n4,n3){$0$}
\drawedge[dash={1.4}{.7},ELdist=2.31](n3,n1){$1$}
\drawedge[dash={1.4}{.7},ELdist=2.31](n0,n3){$1$}
\drawedge[ELdist=1.83](n1,n0){$0$}
\drawedge[dash={1.4}{.7},ELdist=1.19](n1,n2){$1$}
\drawedge[dash={1.4}{.7},ELdist=1.71](n0,n4){$1$}
\drawloop[loopangle=215.88](n3){$0$}
\drawloop[loopangle=90](n0){$0$}
\drawloop[dash={1.4}{.7},ELdist=2.08,loopangle=37.3](n2){$1$}
\drawedge[dash={1.4}{.7},ELdist=2.31](n5,n7){$0$}
\drawedge[dash={1.4}{.7},ELdist=2.21](n7,n9){$0$}
\drawedge[dash={1.4}{.7},ELdist=1.85](n9,n8){$0$}
\drawedge[ELdist=1.86](n8,n6){$1$}
\drawedge[ELdist=1.86](n5,n8){$1$}
\drawedge[dash={1.4}{.7},ELdist=1.83](n6,n5){$0$}
\drawedge[ELdist=1.98](n6,n7){$1$}
\drawedge[ELdist=2.07](n5,n9){$1$}
\drawloop[dash={1.4}{.7},loopangle=90](n5){$0$}
\drawloop[dash={1.4}{.7},loopangle=220.86](n8){$0$}
\drawloop[ELdist=1.84,loopangle=26.57](n7){$1$}
\end{picture}
\caption{Redistribution of tokens after the first move}\label{moves}
\end{center}
\end{figure}

The following observation is immediate.
\begin{lem}
\label{lem:moves}
Suppose that in the game $\Gamma$ played on $\mathscr{A}=(Q,\Sigma,\delta)$, the sequence of chosen symbols forms a word $w\in\Sigma^*$. Then for each $i=1,\dots,n$, the set of states holding the token $\mathbf{i}$ at the end of the game is $q_i.w$.
\end{lem}

Now we express the rules of $\Gamma$ by formulas of propositional logic. For a state $q\in Q$, let $P_0(q)$ and $P_1(q)$ stand for the sets of all preimages of $q$ under the actions of the input symbols 0 and respectively 1, that is, if $a$ is either of the two symbols, $P_a(q)=\{p\in Q\mid q\in p.a\}$. Consider for every $t=1,\dots,\ell$ and all $i,j=1,\dots,n$, the following formulas:
\[
\Psi_{ij}^t:\quad  y_{ij}^t\Longleftrightarrow\Bigl(x_t\wedge\bigvee_{q_k\in P_1(q_j)}y_{ik}^{t-1}\Bigr)\vee \Bigl(\neg x_t\wedge\bigvee_{q_h\in P_0(q_j)}y_{ih}^{t-1}\Bigr).
\]
Observe that the equivalence $\Psi_{ij}^t$ just translates in the language of propositional logic our propagation rule for the tokens that says that the token $\mathbf{i}$ occurs at the state $q_j$ after $t$ moves if and only if one of the following alternatives takes place:
\begin{itemize}
\item the $t$-th move was done with the input symbol 1 and one of the preimages of $q_j$ under the actions of 1 was holding $\mathbf{i}$ after $t-1$ moves, or
\item the $t$-th move was done with the input symbol 0 and one of the preimages of $q_j$ under the actions of 0 was holding $\mathbf{i}$ after $t-1$ moves.
\end{itemize}

\begin{lem}
\label{lem:induction}
For every $t=0,1,\dots,\ell$, every truth assignment $\varphi$ on the set $X$ of letter variables has a unique extension $\ovl{\varphi}$ to the token variables $y_{ij}^s$ that makes the clauses in $C_0$ and the formulas $\Psi_{ij}^s$ hold true $(i,j=1,\dots,n, \ s=1,\dots,t)$. The token variable $y_{ij}^s$ gets value 1 under $\ovl{\varphi}$ if and only if after the moves $\varphi(x_1),\dots,\varphi(x_s)$ of the game $\Gamma$, one of the tokens held by the state $q_j$ is $\mathbf{i}$.
\end{lem}

\begin{proof}
We induct on $t$. The indiction basis $t=0$ is clear: we have to satisfy the clauses in $C_0$ and the only way to satisfy a one-literal clause is to assign value 1 to its only literal. Hence, independently of $\varphi$, we have to set for all $i,j=1,\dots,n$,
\[
\ovl{\varphi}(y_{ij}^0)=\begin{cases}
1&\text{if $i=j$},\\
0&\text{otherwise}.
\end{cases}
\]
Observe that then, in the accordance with the initial setting of the game $\Gamma$, the variable $y_{ij}^0$ gets value 1 exactly when the token held by the state $q_j$ is $\mathbf{i}$.

Now suppose that $t>0$ and there exists a unique way to define $\ovl{\varphi}(y_{ij}^s)$ for all $i,j=1,\dots,n$, $s=0,\dots,t-1$, such that the clauses in $C_0$ and the formulas $\Psi_{ij}^s$ with $i,j=1,\dots,n$ and $s=1,\dots,t-1$ hold true. If the variable $x_t$ is assigned the value $\varphi(x_t)$, the value of the right hand side of each equivalence  $\Psi_{ij}^t$ is uniquely defined, and to make this equivalence hold true, we must assign the value to the left hand side, that is, the variable $y_{ij}^t$.  This gives a unique way to extend $\ovl{\varphi}$ to the variables $y_{ij}^t$, where $i,j=1,\dots,n$. As observed prior to the formulation of the lemma, the equivalences $\Psi_{ij}^t$ express the rule of $\Gamma$. Therefore the token $\mathbf{i}$ will migrate to the state $q_j$ after the move $\varphi(x_t)$ if and only if the variable $y_{ij}^t$ gets value~1 under this extension.
\end{proof}

For each $t=1,\dots,\ell$, we define the set $C_t$ as the set of all clauses of a suitable CNF (conjunctive normal form) equivalent to $\bigwedge\limits_{1\le i,j\le n}\Psi_{ij}^t$. In our basic encoding, the set $C_t$ consists of the following clauses:
\begin{gather}
\label{eq:long}\neg y_{ij}^t\vee x_t\vee\bigvee_{q_h\in P_0(q_j)}y_{ih}^{t-1},\qquad \neg y_{ij}^t\vee \neg x_t\vee\bigvee_{q_k\in P_1(q_j)}y_{ik}^{t-1},\\
\label{eq:short1}y_{ij}^t\vee \neg x_t\vee \neg y_{ik}^{t-1}\  \text{ for each $q_k\in P_1(q_j)$},\\
\label{eq:short2}y_{ij}^t\vee x_t\vee \neg y_{ih}^{t-1}\ \text{ for each $q_h\in P_0(q_j)$}.
\end{gather}
The verification of the equivalence between $\bigwedge\limits_{1\le i,j\le n}\Psi_{ij}^t$ and the conjunction of the clauses in \eqref{eq:long}--\eqref{eq:short2} is routine, and we omit it.

It may be worth explaining how the clauses of the form \eqref{eq:long}--\eqref{eq:short2} are understood in the case when one of the sets $P_0(q_j)$ or $P_1(q_j)$ or both of these sets happen to be empty. In~\eqref{eq:long} the disjunctions over the empty sets are omitted so that if, say, $P_0(q_j)=\varnothing$, then the first clause in \eqref{eq:long} reduces to $\neg y_{ij}^t\vee x_t$. As for~\eqref{eq:short1} or~\eqref{eq:short2}, these clauses disappear whenever $P_1(q_j)$ or, respectively $P_0(q_j)$ are empty.

In order to calculate the number of clauses in $C_t$, denote by $m$ the number of all \emph{transitions in $\mathscr{A}$}, that is, triples $(q,a,q')\in Q\times\Sigma\times Q$ with $q'\in\delta(q,a)$. Clearly, for each fixed $i$, the number $\sum_{j=1}^n(|P_1(q_j)|+|P_0(q_j)|)$ of clauses of the forms~\eqref{eq:short1} and~\eqref{eq:short2} is equal to $m$, whence the total number of such ``short'' clauses is $mn$. As for ``long'' clauses in~\eqref{eq:long}, there are at most two such clauses for each fixed pair $(i,j)$, whence their total number does not exceed $2n^2$.  Altogether, $|C_t|\le n(m+2n)$ for each $t=1,\dots,\ell$.

Lemma~\ref{lem:moves} readily implies that a word $w=a_1\cdots a_\ell$ is $D_3$-synchronizing for $\mathscr{A}$ if and only if after the moves $a_1,\dots,a_\ell$ in the game $\Gamma$ on $\mathscr{A}$, some state $q_j$ holds all tokens $\mathbf{1},\dots,\mathbf{n}$. This is equivalent to saying that the formula
\begin{equation}
\label{eq:synchronization}
\bigvee_{j=1}^n\bigwedge_{i=1}^n y_{ij}^\ell
\end{equation}
holds true under the extension, specified in Lemma~\ref{lem:induction}, of the truth assignment on $X$ defined by $w$. A little difficulty is that a direct conversion of the formula \eqref{eq:synchronization} into a CNF produces $2^n$ clauses. To overcome this difficulty, we use a standard trick for which we need new variables (this is why we introduce synchronization variables). Let $S$ consist of the following $n^2+1$ clauses:
\[
\bigvee_{j=1}^n z_j\ \text{ and }\  \neg z_j\vee y_{ij}^\ell\ \text{ for all }\ i,j=1,\dots,n.
\]
It is easy to see that the set $S$ and the formula \eqref{eq:synchronization} are equisatisfiable; moreover, if $Y=\{y_{ij}^\ell\mid i,j=1,\dots,n\}$ and $Z=\{z_1,\dots,z_n\}$, then every truth assignment on $Y$  that satisfies \eqref{eq:synchronization} can be extended to a truth assignment on $Y\cup Z$ that satisfies $S$, and, conversely, for every truth assignment on $Y\cup Z$ that satisfies $S$, its restriction to $Y$ satisfies \eqref{eq:synchronization}.

The whole set $C=S\cup\bigcup_{t=0}^\ell C_t$ consists of at most $n(m+2n)\ell+2n^2+1$ clauses. The number of transitions in a NFA with $n$ states two input symbols is upper-bounded by $2n^2$, whence $|C|\le 2\ell n^3+2(\ell+1)n^2+1$. Thus, constructing $(V,C)$ from $\mathscr{A}$ takes time polynomial in $n$ and $\ell$.  Summarizing the above discussion, we arrive at the main result of the section.

\begin{thm}
\label{thm:reduction}
An NFA $\mathscr{A}$ has a $D_3$-synchronizing word of length $\ell$ if and only if the instance $(V,C)$ of SAT constructed above is satisfiable, and the construction takes
time polynomial in the size of $\mathscr{A}$ and the value of $\ell$. Moreover, by the construction, there is a 1-1 correspondence between the $D_3$-synchronizing words of length $\ell$ for $\mathscr{A}$ and the restrictions of satisfying assignments of $(V,C)$ to the letter variables.
\end{thm}

\begin{rmk}
{\rm We do not claim that the above reduction of D3W to SAT is optimal. For instance, it is possible to reduce the number of variables by getting rid of the letter variables. Namely, for each pair of $i,j\in\{1,\dots,n\}$ and each $t\in\{1,\dots,\ell\}$, one could take the clause
\begin{equation}
\label{eq:longnew}
\neg y_{ij}^t\vee\bigvee_{q_h\in P_0(q_j)}y_{ih}^{t-1}\vee\bigvee_{q_k\in P_1(q_j)}y_{ik}^{t-1}
\end{equation}
instead of the clauses in \eqref{eq:long} and the set of clauses of the form
\begin{equation}
\label{eq:shortnew}
y_{ij}^t\vee \neg y_{ih}^{t-1}\vee \neg y_{ik}^{t-1}\ \text{ for $h$ and $k$ such that}\ q_h\in P_0(q_j)\ \text{and}\ q_k\in P_1(q_j)
\end{equation}
instead of the ones in~\eqref{eq:short1} and~\eqref{eq:short2}. It is easy to see that \eqref{eq:long} and \eqref{eq:longnew} are equisatisfiable, and so are the sets of clauses in~\eqref{eq:short1},~\eqref{eq:short2} on the one hand and in \eqref{eq:shortnew} on the other.

We have preferred to keep the letter variables because of the fact mentioned in Theorem~\ref{thm:reduction}: if a $D_3$-synchronizing word of length $\ell$ exists, we can immediately recover it from the restriction of a satisfying assignment to the letter variables.}
\end{rmk}

\section{Experimental results}
\label{experimental}
Here we overview our experiments and present some of their results. Our basic procedure has been organized as follows.
 \begin{enumerate}
\item[1.] A positive integer $n$ (the number of states) is fixed. In the experiments which results we report here, we have considered $n\le 100$.
\item[2.] A random NFA $\mathscr{A}$ with $n$ states and 2 input symbols is generated. We have used two models of random generation that are specified below.
\item[3.] We check whether $\mathscr{A}$ has an input symbol whose action is defined at each state. If it is not the case, the NFA $\mathscr{A}$ cannot be $D_3$-synchronizing, and we return to Step~2 to generate another random NFA.
\item[4.] A positive integer $\ell_0$ (the hypothetical length of the shortest $D_3$-syn\-chro\-niz\-ing word for $\mathscr{A}$) is chosen. Initially, we chose $\ell_0$ to be close to $n$ but, as our early experiments have revealed, it is much more practical to start with smaller values of $\ell_0$. We introduce three integer variables $\ell_{\min}$, $\ell$, and $\ell_{\max}$ and initialize them as follows: $\ell_{\min}:=1$, $\ell:=\ell_0$, $\ell_{\max}:=2\ell_0$.
\item[5.] The pair $(\mathscr{A},\ell)$ is encoded into a SAT instance as described in Section~\ref{encoding}.
\item[6.] A SAT solver is invoked to solve the SAT instance obtained in Step~5. We have used MiniSat 2.2.0; see~\cite{Minisat} for a description of the underlying ideas of MiniSat and~\cite{Minisat_page} for a discussion and the source code of the solver.
\item[7.] The binary search on $\ell$ is performed. In more detail, if the SAT solver returns YES on the encoding of the pair $(\mathscr{A},\ell)$, we first check whether or not $\ell=\ell_{\min}$. If $\ell=\ell_{\min}$, then $\ell$ is the length of the shortest $D_3$-syn\-chro\-nizing word for $\mathscr{A}$, and we go to Step~2 to generate another random NFA. If $\ell>\ell_{\min}$, we update the variables $\ell_{\max}$ and $\ell$ by letting
    \begin{gather*}
    \ell_{\max}:=\ell,\quad \ell:=\lfloor\frac{\ell_{\min}+\ell_{\max}}2\rfloor,
    \end{gather*}
    keep the value of $\ell_{\min}$ and go to Step~5. If the SAT solver returns NO on the encoding of the pair $(\mathscr{A},\ell)$, we check whether or not $\ell=\ell_{\max}$. If $\ell=\ell_{\max}$, we interpret this as the evidence that the NFA $\mathscr{A}$ fails to be $D_3$-syn\-chro\-nizing\footnote{Of course, the equality $\ell=\ell_{\max}$ only means that $\mathscr{A}$ has no $D_3$-synchronizing word of length $\le 2\ell_0$, and it is not excluded, in principle, that the NFA is $D_3$-synchronizing but its shortest $D_3$-synchronizing word is very long. However, by suitable preprocessing and choosing an appropriate value of the parameter $\ell_0$, we have got rid of  the ``bad'' cases when the SAT solver returns NO and $\ell=\ell_{\max}$ in our experiments.} and go to Step~2 to generate another random NFA.
    If $\ell<\ell_{\max}$, we update the variables $\ell_{\min}$ and $\ell$ by letting
    \begin{gather*}
    \ell_{\min}:=\ell+1,\quad \ell:=\lceil\frac{\ell_{\min}+\ell_{\max}}2\rceil,
    \end{gather*}
    keep the value of $\ell_{\max}$ and go to Step~5.
\end{enumerate}

We implemented the algorithm outlined above in C++ and compiled with GCC 4.9.2. In our experiments we used a personal computer
with an Intel(R) Core(TM) i5-2520M processor with 2.5 GHz CPU and 4GB of RAM. For each fixed $n$, up to $400$ NFAs that passed Step~3 were analyzed. The average calculation time (for one NFA) was 400 seconds for $n = 30$ and 4350 seconds for $n = 60$.

The two models we used for random generation of an NFA $\mathscr{A}=(Q,\Sigma,\delta)$ with $n$~states and 2 input symbols are the \textit{uniform model} based on the uniform distribution and the \textit{Poisson model} based on the Poisson distribution with some parameter $\lambda$. For each state $q\in Q$ and each symbol $s\in\Sigma$, we first choose a number $k\in\{0,1,2,\dots,n\}$ that serves as the cardinality of the set $\delta(q,s)$. In the uniform model, each $k$ is chosen with probability $\frac1{n+1}$ while in the Poisson model with parameter $\lambda$, each $k<n$ is chosen with probability $e^{-\lambda}\frac{\lambda^{k}}{k!}$ and $n$ is chosen with probability $1-e^{-\lambda}\sum_{k=0}^{n-1}\frac{\lambda^{k}}{k!}$. With $k$ being chosen, we proceed the same in both models, by choosing a $k$-element subset from all $\binom{n}{k}$ subsets of $Q$ with cardinality $k$ uniformly at random and letting $\delta(q,s)$ be the chosen subset.

In each of the two models, it is easy to estimate the fraction of automata that survive Step 3. The corresponding results are stated in the following proposition which proof amounts to straightforward calculations and is therefore omitted.

\begin{pro}
\label{prop:probabality}
The probability that a random NFA with $n$~states and $2$ input symbols has an input symbol whose action is defined at each state is
\begin{equation}
\label{eq:unifom}
 2(1-\frac{1}{n+1})^{n}-(1-\frac{1}{n+1})^{2n}
\end{equation}
if the NFA is generated under the uniform model and
\begin{equation}
\label{eq:poisson}
2(1-e^{-\lambda})^n-(1-e^{-\lambda})^{2n}
\end{equation}
if the NFA is generated under the Poisson model with parameter $\lambda$.
\end{pro}

Observe that as $n$ grows, the expression in \eqref{eq:unifom} tends to $2e^{-1}-e^{-2}\approx 0.600$ while the expression in \eqref{eq:poisson} tends to 0. In the further discussion, we always assume that the NFA considered have passed Step 3.

For the uniform model, our experiments produced results that may seem surprising at the first glance. Namely, it turns out that for an overwhelming majority of NFAs, the length of the shortest $D_3$-synchronizing word is equal to 2, and this conclusion does not depend on the state number $n$, at least within the range of our experiments (recall that we have considered $n\le 100$). For an illustration, see
Figure~\ref{uniform model} in which the horizontal axis is the length of the shortest $D_3$-synchronizing word and
the vertical axis is the number of NFAs. The blue and the yellow circles represent NFAs with 20 and 30 states respectively.\begin{figure}[th]
\includegraphics [width=0.99\textwidth]{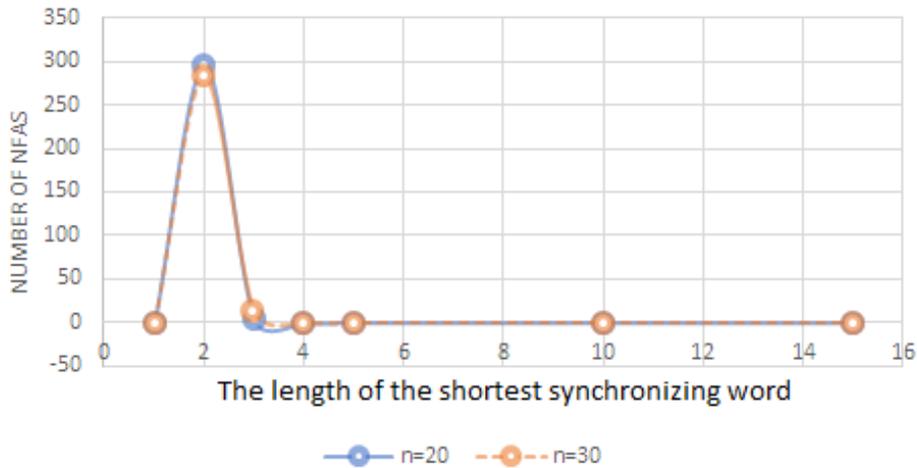}
\caption{Distributions of~20- and 30-state NFAs generated under the uniform model according to the length of their shortest $D_3$-synchronizing words}
\label{uniform model}
\end{figure}

Insofar, we have got no rigorous theoretical explanation of the observed phenomenon. However, even a quick analysis of the uniform model reveals that NFAs it produces should tend to have rather short $D_3$-synchronizing words. Indeed, if an NFA $\mathscr{A}=(Q,\Sigma,\delta)$ with $n$~states and 2 input symbols is generated under the uniform model, then the expected cardinality of the set $\delta(q,s)$ is $\frac{n}2$ for every $q\in Q$ and $s\in\Sigma$. Therefore the expected size of every set of the form $q.w$ with $w\in\Sigma^2$ is close to $n$. Hence it is quite likely that $\bigcap_{q\in Q}q.w\ne\varnothing$ for some word $w$ of length~2, which is then a $D_3$-synchronizing word for $\mathscr{A}$.

Some sample experimental results for the Poisson model are presented in Figure~\ref{n60l5}. The three histograms in Figure~\ref{n60l5} correspond to 60-state NFAs generated under the Poisson models with three different values of the parameter $\lambda$ and demonstrate how these NFAs are distributed according to the length of their shortest $D_3$-synchronizing words. As in Figure~\ref{uniform model}, the horizontal axis is the length of the shortest $D_3$-synchronizing word and the vertical axis is the number of NFAs.
\begin{figure}[p]
 \includegraphics [width=0.95\textwidth]{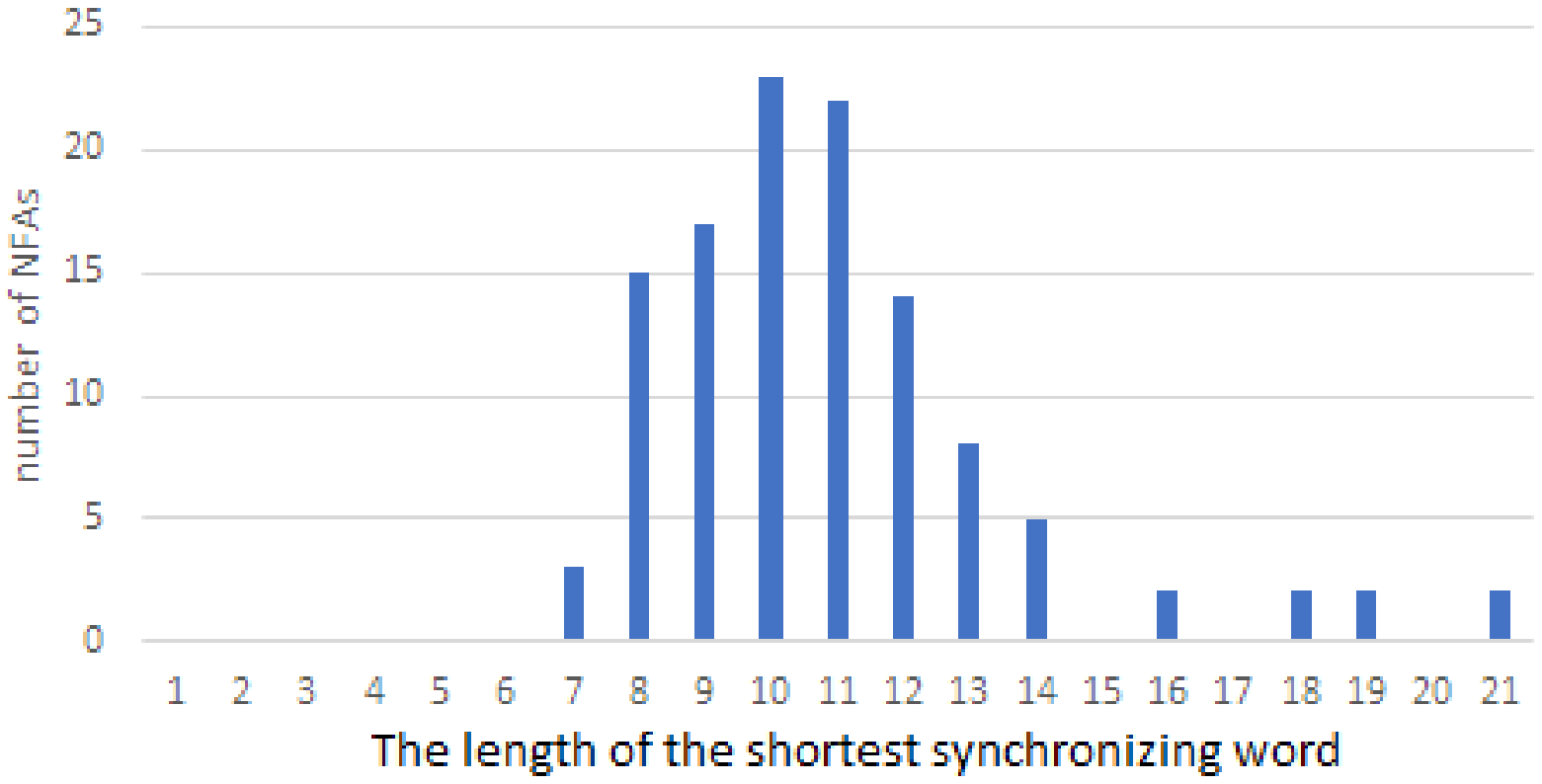}
 \includegraphics [width=0.95\textwidth]{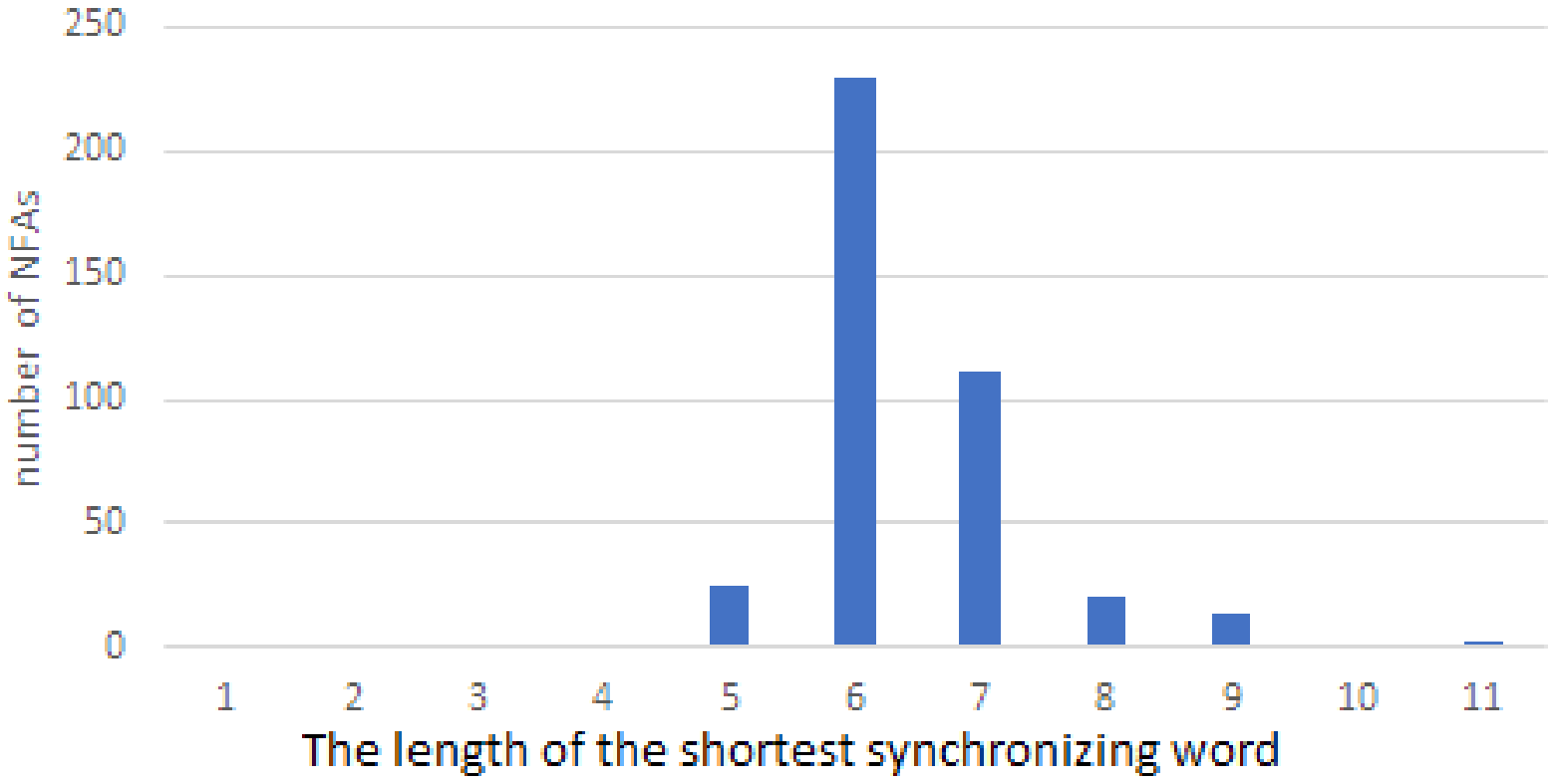}
 \includegraphics [width=0.95\textwidth]{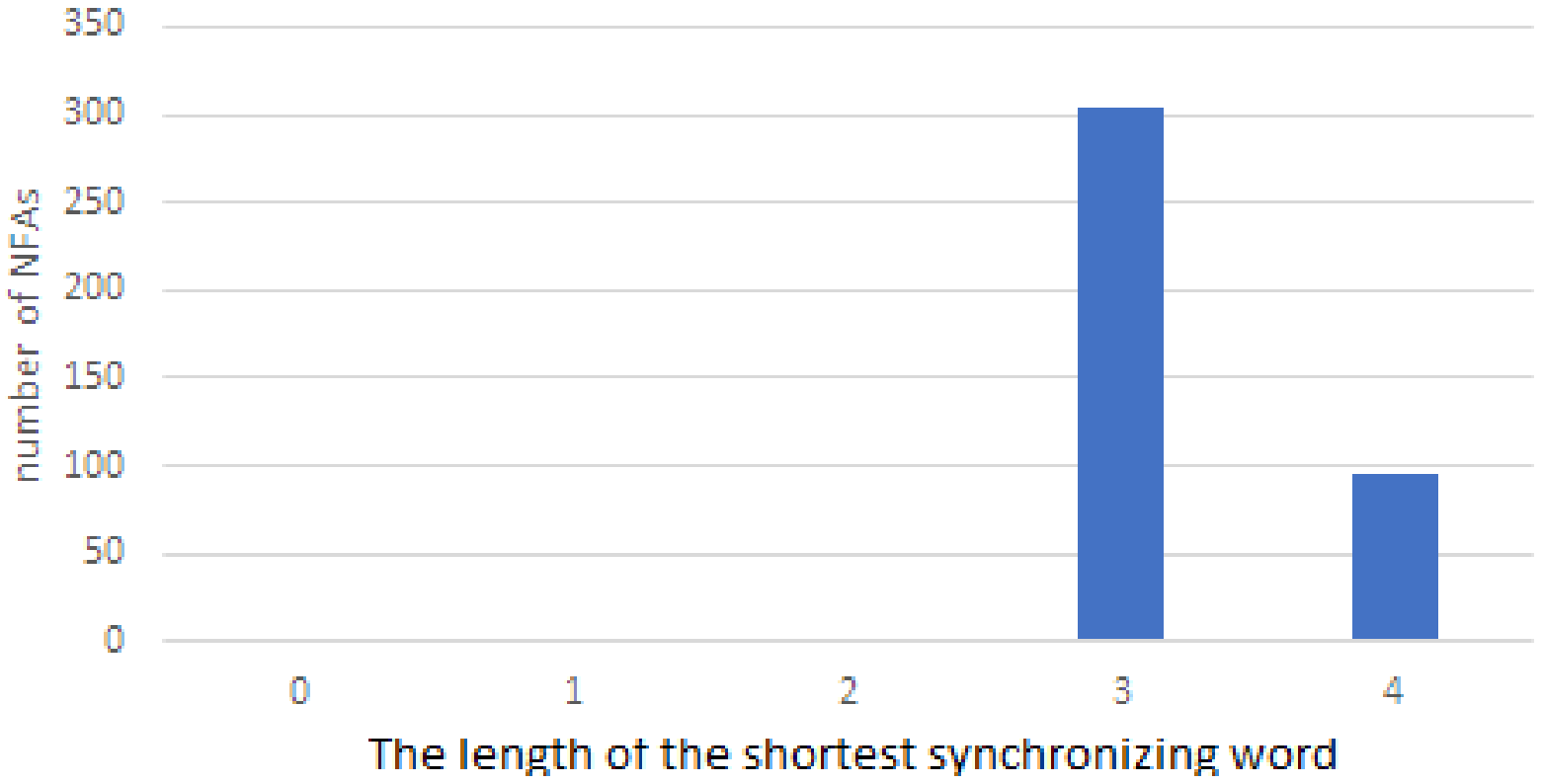}
    \caption{Distributions of~60-state NFAs generated under the Poisson models with $\lambda=1$ (top), $\lambda=2$ (middle), $\lambda=5$ (bottom) according to the length of their shortest $D_3$-synchronizing words}
    \label{n60l5}
\end{figure}

We see that if the number of states is fixed, the expected length of the shortest $D_3$-synchronizing word decreases as the parameter $\lambda$ grows. This can be explained by an informal argument of the same flavour as the reasoning used above to explain the outcome of our experiments with NFAs generated under the uniform model. Indeed, if an NFA $\mathscr{A}=(Q,\Sigma,\delta)$ with $n$~states and 2 input symbols is generated under the Poisson model wiht parameter $\lambda$, it follows from a basic property of the Poisson distribution that $\lambda$ is close to the expected cardinality of sets $\delta(q,s)$ for every $q\in Q$ and $s\in\Sigma$. The larger are these sets, the smaller is the value of $\ell$ such that the expected size of sets of the form $q.w$ with $w\in\Sigma^\ell$ becomes close to $n$.

Our experiments also show that if the parameter $\lambda$ is fixed, the expected length of the shortest $D_3$-synchronizing word grows with the number of states but the growth rate is rather small. For each $n\le 100$, we have calculated the average length $E_1(n)$ of the shortest $D_3$-synchronizing words for $n$-state NFAs generated under the Poisson model with $\lambda=1$. Then, using the method of least squares, we have searched for an explicit function of $n$ that approximates $E_1(n)$ and found the following solution:
\[
E_1(n)\approx (0.57+ 0.66\ln n)^{2}.
\]
For $\lambda=2$, the same procedure has led to the following approximation of the similarly defined quantity $E_2(n)$ calculated from our experimental data:
\[
E_2(n)\approx (0.77+ 0.43\ln n)^{2}.
\]
Similar approximations have been obtained for other values of the parameter $\lambda$.

\section{Conclusion and future work}
\label{final}

We have presented an attempt to approach the problem of computing a $D_{3}$-synchronizing word of minimum length for a given NFA via the SAT-solver method. We think that our results do provide some evidence for this approach to be feasible in principle. Of course, they constitute  only the very first steps, and more work is needed to improve the performance of our implementation and to enlarge its range.

We see several resources for improvements. First of all, we may try to modify the basic encoding described in Section~\ref{encoding}. There are several options for such modifications that all look promising but it is hard to predict a priori which one will prove to be the most efficient, and we have to go through several rounds of trial-and-error. As an example of a relatively successful trial, we briefly report one of the modifications that have already been implemented by the first author.

As mentioned in the description of our basic algorithm in Section~\ref{experimental}, every $D_3$-syn\-chro\-nizing NFA $\mathscr{A}$ must have an everywhere defined input symbol. If all input symbols of $\mathscr{A}$ are everywhere defined, one can use the transformations described in \cite[Lemma~8.3.8]{Ito} or \cite[Section~2]{DonZantema17} to convert $\mathscr{A}$ into a DFA $\mathscr{A}'$ such that $\mathscr{A}$ is $D_3$-syn\-chro\-nizing if and only if $\mathscr{A}'$ is synchronizing and the minimum length of $D_3$-syn\-chro\-nizing words for $\mathscr{A}$ is the same as the minimum length of synchronizing words for $\mathscr{A}'$. Since there are powerful methods to compute shortest synchronizing words for DFAs with up to 350 states (see, e.g., \cite{KKS15}), we can apply one of these methods to $\mathscr{A}'$. Hence, we can restrict ourselves to the case when one of the input symbols of $\mathscr{A}$ is not everywhere defined.

If we consider only NFAs with 2 input symbols, 0 and 1, say, we conclude that we may assume that 0 is everywhere defined while 1 is not. Every $D_3$-syn\-chro\-nizing word for such an NFA should start with the symbol 0. Therefore one can start our solitaire-like game $\Gamma$ described in  Section~\ref{encoding} from the position that arises after the first application of 0, and the  basic encoding can be modified accordingly\footnote{If we re-use the illustrative example in Figure~\ref{moves}, the new initial position for this example will be the one shown in bottom left.}. For an NFA with $n$ states and $m$ transitions, this preprocessing allows one to save $n^2$ variables and around $n^2+2m$ clauses in
the resulting instance of SAT. Our experiments show that this modification indeed reduces the execution time of solving D3W-instances for NFAs with $\ge20$ states, and the average time decrease reaches 50\% for NFAs with $\ge50$ states. Also, the modification has allowed us to solve D3W for NFAs with more than 100 states which size was out of reach with the basic encoding.

Of course, the efficiency of our approach depends not only on the way we encode the problem but also on software and hardware used in the implementation. Besides optimizing our own code, we have plan to experiment with more advanced SAT-solvers, namely, with CryptoMiniSat \cite{Crypto} and lingeling \cite{lingeling}. Using more powerful computers constitutes yet another obvious direction for improvements. In particular, our approach is clearly amenable to parallelization since calculations needed for different automata are completely independent so that in principle, we can work in parallel with as many automata as many processors are available.

Our future work should include theoretical explanations for phenomena observed in our experiment as well as extending our study to automata with arbitrarily many input symbols and to other versions of NFA synchronization such as $D_1$- and $D_2$-synchronization mentioned in Section~\ref{sec:intro}.

\end{document}